\documentclass[graybox]{svmult}

\usepackage{mathptmx}       
\usepackage{helvet}         
\usepackage{courier}        
\usepackage{type1cm}        
\usepackage{makeidx}         
\usepackage{graphicx}        
\usepackage{multicol}        
\usepackage[bottom]{footmisc}

\usepackage{amsmath,amssymb,amsfonts,bbm,mathrsfs}

\usepackage{changes}

\newcommand{\ot}[0]{\otimes}

\newcommand{\ensm}[1]{\ensuremath{#1}}

\renewcommand{\alph}{\ensm{\alpha}}

\newcommand{\bra}[1]{\ensm{\langle #1|}}

\newcommand{\ket}[1]{\ensm{| #1 \rangle}}

\newcommand{\be}{\begin{equation}}
\newcommand{\ee}{\end{equation}}
\newcommand{\bea}{\begin{eqnarray}}
\newcommand{\eea}{\end{eqnarray}}

\newcommand{\ev}[1]{\langle #1 \rangle}

\newcommand\tr{\mbox{tr}}
\newcommand\one{\mathbb{I}}

\def\one{\mathbbm{1}}
\def\H{{\cal H}}

\makeindex             

\begin{document}

\title*{Guess your neighbour's input: no quantum advantage but an advantage for quantum theory}
\titlerunning{Guess your neighbour's input ...}

\author{Antonio Ac\'in, Mafalda L. Almeida, Remigiusz Augusiak, and Nicolas Brunner}
\authorrunning{A. Ac\'in, M. L. Almeida, R. Augusiak, and N. Brunner}

\institute{A. Ac\'in \at ICFO--Institut de Ci\`encies
Fot\`oniques, 08860 Castelldefels (Barcelona), Spain and
ICREA--Instituci\'o Catalana de Recerca i Estudis Avan\c{c}ats,
08010 Barcelona, Spain, \email{antonio.acin@icfo.es} \and M. L.
Almeida \at Centre for Quantum Technologies, National University
of Singapore, 3 Science drive 2, Singapore 117543,
\email{cqtmla@nus.edu.sg} \and R. Augusiak \at ICFO--Institut de
Ci\`encies Fot\`oniques, 08860 Castelldefels (Barcelona), Spain,
\email{remigiusz.augusiak@icfo.es} \and N. Brunner \at H.H. Wills
Physics Laboratory, University of Bristol, Bristol, BS8 1TL,
United Kingdom, \email{N.Brunner@bristol.ac.uk}}

\maketitle

\abstract{Quantum mechanics dramatically differs from classical
physics, allowing for a wide range of genuinely quantum phenomena.
The goal of quantum information is to understand information
processing from a quantum perspective. In this mindset, it is thus
natural to focus on tasks where quantum resources provide an
advantage over classical ones, and to overlook tasks where quantum
mechanics provides no advantage. But are the latter tasks really
useless from a more general perspective? Here we discuss a simple
information-theoretic game called 'guess your neighbour's input',
for which classical and quantum players perform equally well. We
will see that this seemingly innocuous game turns out to be useful
in various contexts. From a fundamental point of view, the game
provides a sharp separation between quantum mechanics and other
more general physical theories, hence bringing a deeper
understanding of the foundations of quantum mechanics. The game
also finds unexpected applications in quantum foundations and
quantum information theory, related to Gleason's theorem, and to
bound entanglement and unextendible product bases.}

\section{Introduction}
Quantum theory is arguably the most accurate scientific theory
designed so far. However, despite this success, we still lack a
deep understanding of the foundations of the theory. An important
goal in the foundations of quantum mechanics is therefore to
recover quantum theory from alternative sets of axioms, motivated
by physical principles rather than mathematical ones \cite{hardy}.

In particular, a case that attracted considerable attention
recently is that of quantum nonlocal correlations. Quantum
nonlocality \cite{bell}, a valuable resource for information
processing \cite{ekert,CC,diqkd,dirng}, is the strongest
manifestation of quantum correlations; distant observers
performing local measurements on a shared entangled state, may
observe correlations between their measurement outcomes which
could provably not been obtained in any local theory. The strength
of quantum correlations appears however to be limited, in a way
that cannot be yet explained by any physical principle. Consider
for instance the principle stating that information cannot be
transmitted instantaneously, the so-called no-signaling principle.
Although this principle is satisfied by all quantum correlations,
preventing from a direct conflict with relativity, it does not
single out quantum correlations. Indeed there exist no-signaling
correlations which are stronger than those allowed in quantum
mechanics \cite{PR}, usually referred to as super-quantum
correlations. Why such correlations would be unlikely to exist in
nature and whether there exist a physical principle singling out
quantum correlations are important issues in the foundations of
quantum mechanics \cite{brassardNP,popescuNP,bub}.

Several approaches have been investigated to discuss this problem.
The first consists in investigating the capabilities for
information processing of super-quantum correlations, and to
compare them with that of quantum correlations. Interestingly it
was shown that the availability of certain super-quantum
correlations, instead of quantum correlations, would tremendously
increase the communication power of classical communication. In
particular, it was shown that some of them would collapse
communication complexity \cite{Dam2005,Brassard2006,Brunner2009}
(hence dramatically reducing the amount of classical communication
required to solve a large class of problems \cite{CC}) or violate
the principles of information causality
\cite{Pawlowski2009a,Allcock2009} and macroscopic locality
\cite{ML}. A second approach, perhaps less demanding, starts from
assuming 'local quantum mechanics'. In other words the statistics
of local measurements are assumed to follow Born's rule. What
other principle should then be imposed in order for the global
statistics to be quantum? In the bipartite case, it is proven that
the no-signaling principle is enough to single out quantum
correlations \cite{Barnum,Universal}. Importantly, while both of
these approaches have proven to be (at least partially) successful
in the case of two parties, none of them can tackle the general
multipartite scenario.

Here we present 'Guess your neighbour's input' (GYNI) \cite{GYNI},
a simple multipartite game, the rules of which can be understood
intuitively from its name. Despite its innocuous appearance, the
game captures crucial features of multipartite quantum
correlations. The main aspect of the game is the following.
Whereas players sharing quantum resources do not have any
advantage over players sharing classical resources, it turns out
that players sharing super-quantum correlations have an advantage
over players sharing either classical or quantum resources. In
other words, the limitation of quantum resources is here not a
mere consequence of the no-signaling principle. Hence, the game of
GYNI provides a natural separation between quantum and
super-quantum correlations. More generally these results point
towards a strengthening of the no-signaling principle, in the
general multipartite case, obeyed by quantum mechanics. Therefore,
whereas the game of GYNI may seem a priori useless from a quantum
perspective, it does in fact bring a novel and fresh perspective
on the foundations of quantum theory \cite{andreas}.

Although it is not clear yet what fundamental principle lies
behind the quantum limitations for GYNI, several important
features of such a principle can already be identified. In
particular, this principle must be genuinely multipartite, which
can be shown directly from the GYNI game. This is because there
exist multipartite super-quantum correlations, that will
nevertheless satisfy any bipartite principle \cite{gwan} (see also
\cite{singapore}). Hence quantum correlations can provably not be
recovered from any principle that is inherently bipartite (such as
no trivial communication complexity or information causality).

Moreover, it can be shown, using GYNI, that there exist
multipartite super-quantum correlations obeying the Born rule
locally \cite{Universal}. Therefore, in the multipartite case, the
no-signaling principle is not enough to recover quantum
correlations from local quantum mechanics. This result also has
fundamental consequences on extensions of Gleason's theorem
\cite{gleason} to composite systems.

Finally, we shall see that GYNI has also applications beyond
quantum foundations. In particular, the game turns out to be
strongly related \cite{BellUPB1,BellUPB2} to topics of quantum
information theory, namely bound entanglement
\cite{HorodeckiBound} and unextendible product bases
\cite{BennettUPB}. This is surprising since these subjects seem to
be completely unconnected at first sight. This connection deepens
our understanding of Bell inequalities with no quantum advantage.
In particular it allows us to derive such inequalities from
unextendible product bases.

This chapter is structured as follows. In Section \ref{sec:gyni},
after giving a brief background introduction to nonlocal
correlations, we present the GYNI game and derive the winning
probabilities for various types of correlations (local, quantum,
and no-signaling). Applications of GYNI are presented in Sections
\ref{sec:gleason} and \ref{sec:qcorr}. First, in Section
\ref{sec:gleason}, we discuss results on the extension of
Gleason's theorem for composite systems. Then, in Section
\ref{sec:qcorr}, we shall see that any information-theoretic
principle catpuring quantum correlations must be genuinuely
multipartite. In Section \ref{sec:upb}, after presenting in detail
the connection between GYNI and unextendible product bases, we
will make use of this connection to go beyond GYNI, and to better
understand the structure of Bell inequalities with no quantum
advantage. Finally, we will conlude in Section
\ref{sec:conclusion}.

\section{Guess your neighbour's input}
\label{sec:gyni}

\subsection{Background: classical, quantum and no-signalling correlations}

The definition of (non)locality was introduced by Bell, as a
rigorous physical and mathematical framework to test the
Einstein-Podolsky-Rosen paradox. Consider two distant observers,
Alice and Bob, sharing a physical system, and performing local
measurements on their subsystems. Alice and Bob's choice of
observables are labeled by $x_1$ and $x_2$ respectively, and take
outcomes $a_1$ and $a_2$ (hereafter the subscripts will be
omitted). The joint probability distribution of outcomes,
conditioned on the choice of observables, is represented by
$P(a_1,a_2|x_1,x_2)$. This set of data is described as
\textit{local} (or classical) if and only if $P(a_1,a_2|x_1,x_2)$
can be reproduced by a local hidden-variable model\footnote{By
simplicity, we consider $\lambda$ to be discrete, but all the
formulation can be extended to the continuous case.}, that is, iff
it can be written in the form
\begin{equation}\label{localcorr}
P_L(a_1,a_2|x_1,x_2)=\sum_\lambda P(\lambda) P(a_1|x_1,\lambda)P(a_2|x_2,\lambda)\,.
\end{equation}
Here individual outcomes are completely specified by the choice of
local observable and the shared (hidden) variable $\lambda$.
Indeed, Alice and Bob's outcomes may be correlated via the
hidden-variable $\lambda$, which is distributed with probability
density $P(\lambda)$.

The probability distribution $P(a_1,a_2|x_1,x_2)$ is said realizable in quantum mechanics (or in short, to be quantum) if and only if it can be written in the following form:
\begin{equation}\label{quantumcorr}
P_Q(a_1,a_2|x_1,x_2)=\tr(\rho_{AB}M_{a_1}^{x_1}\otimes M_{a_2}^{x_2}),
\end{equation}
where the state of system $\rho_{AB}$ is defined by a density
operator on the joint Hilbert space $\H_A\otimes \H_B$, and
$M_{a_1}^{x_1}$, $M_{a_2}^{x_2}$ are the local generalized
measurements (positive semidefinite operators on the local Hilbert
space such that $\sum_{a_j} M_{a_j}^{x_j}=\one$ $(j=1,2)$ with
$\mathbbm{1}$ denoting an identity matrix of the dimension
following from the context). Indeed quantum correlations are
stronger than classical ones, hence there exist quantum
distributions $P_Q(a_1,a_2|x_1,x_2)$ which cannot be written in
the form \eqref{localcorr}.

A crucial feature of both classical and quantum correlations is
that they satisfy the no-signalling principle: instantaneous
information transmission is impossible. More formally the
principle says that Alice's measurement outcome is uncorrelated to
Bob's choice of measurement, that is
\begin{equation}\label{nscorr}
\underset{x_2,x_2'}{\displaystyle\forall}\quad \sum_{a_2} P_{NS}(a_1,a_2|x_1,x_2)= \sum_{a_2} P_{NS}(a_1,a_2|x_1,x_2')\equiv P_A(a_1|x_1).
\end{equation}
Similar equations must be obeyed for Bob's marginal distribution.
Correlations satisfying this principle, as well as normalization
and positivity, are referred to as nonsignaling correlations
\cite{barrett}. Interestingly, there exist nonsignaling
correlations that are not quantum \cite{PR}, i.e. cannot be
written in the form \eqref{quantumcorr}.

The above definitions are naturally generalized to the
multipartite case: local correlations between $N$ parties are
described by
\begin{equation}\label{mlocalcorr}
P_L(a_1,\ldots,a_N | x_1,\ldots,x_N)=\sum_\lambda P(\lambda) P(a_1|x_1,\lambda)P(a_2|x_2,\lambda)\ldots P(a_N|x_N,\lambda),
\end{equation}
Quantum correlations are given by
\begin{equation}\label{quantumcorr}
P_Q(a_1,...,a_N|x_1,...,x_N)=\tr(\rho M_{a_1}^{x_1}\otimes ... \otimes M_{a_N}^{x_N}),
\end{equation}
where $\rho$ denotes the quantum state shared between the parties.
Finally nonsignalling correlations are defined such that no party
is allowed to signal to others through his choice of measurement,
that is
\begin{equation}\label{mnscorr}
\underset{j,j'}{\displaystyle\forall} \quad\sum_{a_j}P_{NS}(a_1,\ldots,a_N|x_1,\ldots,x_N)=\sum_{a_j'}P_{NS}(a_1,\ldots,a_N|x_1,\ldots,x_N).
\end{equation}
and similar relations for any two sets of parties.

In order to distinguish between these three kinds of correlations
(local, quantum, and nonsignaling) one devises a Bell test,
involving a certain number (usually finite) of parties,
observables and outcomes. It is convenient to represent a
probability distribution $P(a_1,...,a_N|x_1,...,x_N)$ as a vector
of probabilities $\mathbf{P}$, with entries $P( \mathbf a |
\mathbf x)= P(a_1,...,a_N|x_1,...,x_N)$. In this vector space,
Bell inequalities are given by linear expressions
\begin{equation}\label{bellineq}
S= \sum_j \alpha_j \mathbf P_j\leq \omega_c\,.
\end{equation}
The coefficients $\alpha_j$ are real. The bound of the inequality,
i.e. $\omega_c$, is the largest value of the Bell polynomial $S$
for any local probability distribution, i.e. of the form
\eqref{mlocalcorr}. The set of local correlations defines a convex
polytope. Hence it can be described by a finite set of linear
inequalities, that are called tight Bell inequalities.

The local set is a strict subset of the set of quantum
correlations. The latter is still a convex set, although no longer
a polytope. It can, however, be described by an infinite set of
quantum Bell inequalities, similar to $\eqref{bellineq}$ but
replacing the classical bounds by quantum ones, $\omega_q$, which
may in general exceed the classical one, i.e. $\omega_q\geq
\omega_c$.

Finally, the set of no-signalling correlations is also a convex
polytope, which is strictly larger than the quantum set. Its
facets are given by positivity inequalities, stating that joint
probabilities are positive. The largest value of a Bell polynomial
$S$ for any no-signaling probability distribution is denoted
$\omega_{ns}$; indeed, in general $\omega_{ns} \geq \omega_q$.

The scene being set, let us bring in the protagonists.

\subsection{The GYNI game}
Consider $N$ players disposed on a ring. The game starts with each
player receiving a (private) input bit $x_i$ (say from a referee),
distributed according to the probability density $q(\mathbf{x})$.
Now, the name of the game says it all: the goal is that each
player makes a correct guess $a_i$ of his (say) right-hand side
neighbour's input bit (see Fig.~1), that is
\begin{equation}
\underset{i}{\displaystyle\forall}\quad a_i=x_{i+1},
\end{equation}
where $x_{N+1}\equiv x_1$.
Importantly, the players are successful if and only if all the parties make a correct guess.

\begin{figure}[h]
\sidecaption
  \includegraphics[height=0.25\columnwidth,trim=200 250 200 200]{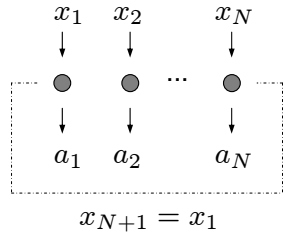}
  \caption{The GYNI game. The goal is that each party outputs its right-neighbour's input: $a_i=x_{i+1}$.}
\end{figure}

The winning probability is defined as
\begin{equation}\label{gyni_ineq}
\omega=\sum_{\mathbf x} q(\mathbf x) P(\mathbf a_{i}=\mathbf x_{i+1}|\mathbf x_i)
\end{equation}
with $ P(\mathbf a_{i}=\mathbf x_{i+1}|\mathbf
x_i)=P(a_1=x_2,\ldots, a_N=x_1|x_1,\ldots, x_N)$. Note that no
communication between the players is allowed during the game.
However, during the preparation stage of the game, the players are
informed of the distribution $q(\mathbf{x})$ of the inputs. They
are allowed to establish a common strategy, which will consist in
utilizing in a judicious way physical resources they are allowed
to share. Here our aim will be to find out how good the parties
can perform at the game when sharing respectively classical,
quantum, and no-signaling correlations. Formally, the game
represents a multipartite Bell test, and eq. $\eqref{gyni_ineq}$
has the structure of a multipartite Bell inequality (see
$\eqref{bellineq}$). Hence our goal will be to determine the
bounds $\omega_c$, $\omega_q$ and $\omega_{ns}$, corresponding to
the classical, quantum and no-signalling bounds of the GYNI Bell
inequality.

\subsection{No quantum advantage}
A central features of the GYNI game is that the maximum winning
probability in the quantum world is exactly the same as in a
classical one. In other words, the GYNI inequalities
\eqref{gyni_ineq} have the same classical and quantum bound, i.e.
$\omega_c=\omega_q$, for any distribution of inputs $q(\mathbf
x)$.

\noindent\textbf{Classical bound.} Let us start by analyzing the
best classical performance. Any probabilistic classical strategy
(which includes the use of shared randomness), can be decomposed
into a convex sum of deterministic strategies. This means that
players can achieve the best winning probability $\omega_c$ by
making a definite guess $a_i$ for each input bit $x_i$. Hence it
is enough to analyze such cases. Imagine that their deterministic
strategy allows them to succeed when receiving some input string
$\mathbf y$, i.e. $a_i(y_i)=y_{i+1}, \forall i$. The input strings
have an interesting orthogonality property: for any other input
$\mathbf x\neq \mathbf y, \bar{\mathbf y}$, there is some $i$ such
that $x_i=y_i$ and $x_{i+1}\neq y_{i+1}$. Then, for any
input-strings $x$, there is always some player $i$ which will make
a wrong guess. He will receive the bit $x_i=y_i$, and output
$a_i(y_i)=y_{i+1}$ according to the strategy, while the correct
would be $a_i=\bar{y}_{i+1}$. It is still possible to score when
receiving $\bar{\mathbf y}$ by setting the strategy to
$a_i(\bar{y_i})=\bar{y}_{i+1}, \forall i$. The best classical
winning probability is then
\begin{equation}\label{cbound}
\omega_c=\max_{\mathbf x}[q(\mathbf x)+q(\bar{\mathbf x})]\,,
\end{equation}
achieved by using $\mathbf y$ such that $q(\mathbf
y)+q(\bar{\mathbf y})=\max_{\mathbf x}[q(\mathbf x)+q(\bar{\mathbf
x})]$.

\noindent\textbf{Quantum bound.} If players have access to quantum
systems, the most general protocol involves a quantum state $\rho$
of arbitrary Hilbert space dimension and general quantum
measurements $M_{x_i}^{a_i}$ corresponding to a probability
distribution
\begin{equation}
P(a_1,\ldots,a_N|x_1,\ldots,x_N)=\tr(\rho M_{x_1}^{a_1}\otimes\ldots\otimes M_{x_N}^{a_N}).
\end{equation}
The best quantum winning probability is then the maximum expected
value of the Bell operator
\begin{equation}
\omega_q=\max_{\psi, \text{meas}}\sum_{\mathbf x} q(\mathbf x) \ev{M_\mathbf{x}}\,.
\end{equation}
where $M_\mathbf{x}\equiv M_{x_1}^{x_2}\otimes\ldots\otimes
M_{x_N}^{x_1}$. Notice that it is enough to optimize over pure
states $\ket{\psi}$ and projective measurements
$M_{a_i}^{x_i}M_{a'_i}^{x_i}=\delta_{a_i=a'_i}M_{a_i}^{x_i}$,
since there are no restrictions on the size of local Hilbert
spaces. Following a similar reasoning to the classical case, take
projectors $M_\mathbf y$ and $M_\mathbf x$, where $ \mathbf x\neq
\mathbf y,  \bar {\mathbf y}$. Then there is some local projector
$i$, defined on the same basis $x_i=y_i$, but projecting on
orthogonal subspaces $x_{i+1}\neq y_{i+1}$. Consequently, the
measurement projectors also obey an orthogonality condition,
\begin{equation}
M_\mathbf{y}M_\mathbf{x}=0 \quad\text{if} \quad \mathbf x\neq \mathbf y, \bar {\mathbf y}\,.
\end{equation}
This property is sufficient to show that
\begin{equation}
\sum_{\mathbf x} q(\mathbf x) \ev{M_\mathbf{x}}\leq\max_{\mathbf x}[q(\mathbf x)+q(\bar{\mathbf x})]\, ,
\end{equation}
which proves that the best quantum winning probability is the same
as the classical one
\begin{equation}
\omega_q=\omega_c\,.
\end{equation}

Indeed, the derivation of the best winning probabilities, in both
the classical and quantum case, relies on a rather natural
orthogonality property (either of deterministic local strategies,
or of orthogonal measurement projectors). Interestingly such a
property is not a consequence of the no-signaling, and does in
general not hold for no-signaling correlations, as we shall see in
the next section.

\subsection{No-signalling advantage}
The game of GYNI is in some sense clearly related to the notion of
signaling. Indeed, if all players can guess correctly their
input's neighbour with a high probability, this will lead to
signaling. Hence it may come to no surprise that quantum
resources, which are indeed no-signaling, give no advantage for
GYNI. Surprisingly this intuition is not correct, as we shall see
here, since certain super-quantum no-signaling correlations can in
fact provide an advantage compared to classical correlations.

\subsubsection{Correlated inputs}
Consider a particular version of the GYNI game in which the inputs
are correlated in the following way: $q(\mathbf x)$ is uniform on
the set of inputs that satisfy the parity condition:
\begin{equation}
\label{promise} q(\mathbf{x})=\left\{\begin{array}{ll}
1/2^{N-1} &\text{if } x_1\oplus\cdots\oplus x_{\hat N}=0\\
0 & \text{otherwise}\,,\end{array}
\right.
\end{equation}
where  $\hat N=N$ if $N$ is odd and $\hat N=N-1$ if $N$ is even.
Using Eq.~\eqref{cbound}, it is easy to check that in classical or
quantum theory, the success probability is limited by
$\omega_c=1/2^{N-1}$. We will see that, allowing for super-quantum
correlations, this limit can be beaten: the best winning
probability $\omega_{ns}$ is upper-bounded by
$\omega_{ns}\leq1/3$. Unlike the previous example, here, although
each party still has absolute uncertainty about his neighbour's
input, no-signalling correlations are able to exploit a global
correlation (the parity of the input-string) to increase the
chance of correct guess.

\noindent\textbf{3-player game.} Let us first consider the
simplest game, featuring three players\footnote{Note that for 2
players, no-signaling correlation provide no advantage.}. The GYNI
inequality is then simply given by
\begin{equation} \label{3player}
\omega=\frac{1}{4}\left[P(000|000)+P(110|011)+P(011|101)+P(101|110)\right] \leq \frac{1}{4},
\end{equation}
where the bound holds for any local or quantum strategy.

Let us first derive an upper bound on the no-signaling winning
probability. Consider the first three terms in \eqref{3player}.
The no-signalling principle implies that
\begin{eqnarray}\label{3NS bound}
P(000|000)\leq \sum_{a_3}P(00a_3|000)=\sum_{a_3}p(00a_3|001)\,,\nonumber\\
P(110|011)\leq \sum_{a_2}P(1a_20|011)=\sum_{a_2}p(1a_20|001)\,,\\
P(011|101)\leq \sum_{a_1}P(a_111|101)=\sum_{a_1}p(a_111|001)\,.\nonumber
\end{eqnarray}
From normalization, we know that the sum of these terms satisfies
$P(000|000)+P(110|011)+P(011|101)\leq1$. We apply a similar
reasoning to the remaining combinations of three probability terms
of Eq.~\eqref{3player}, such that we get
\begin{equation}
3[P_{NS}(000|000)+P_{NS}(110|011)+P_{NS}(011|101)+P_{NS}(101|110)]\leq4\,.
\end{equation}
Hence we obtain an upper limit on the no-signalling winning
probability: $\omega_{ns}\leq1/3$. From this derivation, we also
conclude that it is only possible to reach this limit if every
probability term in the GYNI inequality \eqref{3player} has the
value $1/3$.

Now, it turns out that this upper bound can be reached by an
actual no-signaling probability distribution. The latter is rather
complicated (see \cite{GYNI}), but it would be interesting to
better understand its structure. To be complete, let us mention
that there exist two (among 45) inequivalent classes of extremal
tripartite no-signaling boxes \cite{Pironio2011}, that reach the
best winning no signaling probability $\omega_{ns}=1/3$.

Finally note an interesting feature of inequality \eqref{3player}.
It is a tight Bell inequality, that is, it defines a facet of the
polytope of local correlations \cite{sliwa}. Hence it identifies a
portion of the quantum boundary which is of maximal dimension
\cite{GYNI}.

\noindent\textbf{N-player game.} Next let us consider the general
case of $N$ players, using the condition \eqref{promise} on the
inputs. For any $N$, no-signaling correlations provide an
advantage. To show this, we prove that resources that provide a
winning probability $\omega/\omega_{c}$, in the game with $N$
players, can provide at least the same ratio $\omega/\omega_{c}$
for $N+1$ players. The strategy is very simple: players 1 to $N$
play exactly as in the $N$-player game, while player $N+1$ outputs
his input, $a_{N+1}=x_{N+1}$. This guess is correct when
$x_{N+1}=x_1$, which happens with probability 1/2. Since
$\omega_c(N+1)=(1/2)\omega_c(N)$, the ratio remains the same:
\begin{equation}
\frac{\omega}{\omega_c}(N)=\frac{\omega}{\omega_c}(N+1)\,.
\end{equation}
Then, for any $N\geq 3$, the best no-signalling sucess probability
is at least as good as $(4/3)\omega_c$. This lower bound is
achieved if the first 3 players use the optimal no-signalling
strategy for the 3-player game, while the remaining output their
inputs. They can however do better: using linear programming, we
obtained that $\omega_{ns}/\omega_c=4/3$, for $N=4$;
$\omega_{ns}/\omega_c=16/11$, for $N=5,6$; and
$\omega_{ns}/\omega_c=64/42$, for $N=7,8$. Basing on these three
values the rough estimation would suggest that the ratio
$\omega_{ns}/\omega_c$ scales with $N$ as $4^k/\{(1/3)[(23/3)
4^{k-1} + k+1/3]\}$, where $k=\lfloor(N-1)/2\rfloor$. This in the
limit of $N\to \infty$ gives $\omega_{ns}/\omega_c\to 36/23$.

Remarkably, it turns out that the $N$-partite GYNI Bell
inequalities (with promise \eqref{promise}), hereafter referred to
as GYNI$_N$, are tight for an arbitrary odd $N$~\cite{BellUPB2}
and for $N=4,6$~\cite{GYNI}. It is conjectured that they are tight
for any $N$.

\subsubsection{Upper bounds on $\omega_{ns}$} From the winning probability
in the classical case (Eq.~\eqref{cbound}), we know that
$q(\mathbf x)\leq \omega_c$ for any $\mathbf x$, from which we get
the bound $\omega\leq\omega_c\sum_{\mathbf x} P(\mathbf
a_{i}=\mathbf x_{i+1}|\mathbf x_i)$. Something more meaningful is
obtained if we now assume the distributions to be no-signalling.
Take the summation $\sum_{\mathbf x}P(\mathbf a_{i}=\mathbf
x_{i+1}|\mathbf x)$.  Repeatedly applying the no-signalling
condition \eqref{nscorr}, (first to party $N$, then to $N-1$ and
so on), we get
\begin{multline}
\sum_{x_1, \dots, x_N} P_{NS}(x_2,\ldots,x_N,x_1|x_1,\ldots, x_{N-1}, x_N)\\
\leq \sum_{x_1, \dots, x_N} P_{NS}(x_2,\ldots,x_{N}|x_1,\ldots, x_{N-1})\\
= \sum_{x_1, \dots, x_{N-1}} P_{NS}(x_2,\ldots,x_{N_1}|x_1,\ldots, x_{N-2})
=\dots= 2\,.
\end{multline}
We conclude that the success probability within no-signalling theories is bounded by
\begin{equation}\label{upperns}
\omega_{ns}\leq2\omega_c\,,
\end{equation}
which means that, in general, no-signalling correlations do not
allow deterministic success. As we could predict, for some input
distributions, perfect guessing is only possible if players
communicate. In those cases, it is reasonable to expect that
classical, quantum and no-signalling resources provide exactly the
same best performance.

\subsubsection{Completely uniform distributions of inputs} The
counter-example for the previous intuition is the following: the
completely uniform distribution over the inputs, i.e. $q(\mathbf
x)=1/2^N$. We obtain a tight upper bound on $\omega_{ns}$ by
noticing that $2q(\mathbf x)=\omega_c$, which leads to
\begin{equation}
\omega_{ns}=\frac{\omega_c}{2}\sum_{\mathbf x} P_{NS}(\mathbf a_{i}=\mathbf x_{i+1}|\mathbf x_i)\leq \omega_c\,.
\end{equation}
Classical and no-signalling resources provide exactly the same
best winning probability, in a situation where each player has, a
priori, no information about the input of its neighbour.

Once the GYNI Bell inequality has been introduced, we discuss in
the next sections the application of this inequality in two
different contexts, related to the characterization of quantum
correlations.

\section{Application 1: Gleason's theorem for multipartite systems}
\label{sec:gleason}

Gleason's Theorem~\cite{gleason} is a celebrated theorem in the
foundations of quantum mechanics that allows recovering the Born
rule for quantum probabilities from the structure of quantum
measurements. Recall that a quantum measurement acting on a
Hilbert space of dimension $d$ corresponds to a set of $k$
positive operators, $M_i\geq 0$ with $i=1,\ldots,k$ such that
$\sum_i M_i=\one$. Gleason's Theorem aims at characterizing maps
from quantum measurements to probability distributions. The maps
$\Lambda$ have to satisfy the following properties:

\begin{enumerate}
    \item For any positive operator $0\leq M\leq\one$ one has
    $\Lambda(M)\geq 0$.
    \item Given a quantum measurement, that is, given a set of $k$
    positive operators summing up to the identity, one has
    \begin{equation}\label{glcond}
    \sum_{i=1}^k \Lambda(M_i)=1 .
    \end{equation}
\end{enumerate}

Note that the considered maps are non-contextual, as the
measurement operators are mapped into probabilities independently
of the structure of the measurement they belong to.

Gleason's Theorem implies that all maps satisfying the two
requirements 1 and 2 can be written as $\Lambda(M)=\tr(\rho M)$
for a given quantum state $\rho$, that is, $\rho$ is a positive
operator of trace one. We sketch here the idea of the proof, while
its detailed version may be found e.g. in Ref.~\cite{Busch}.
Notice, however, that the author of \cite{Busch} imposes an
additional condition on $\Lambda$ which, as we show below, can be
simply inferred from 1 and 2. Indeed, consider two measurement
operators $M_1,M_2$ such that $M_3=\one-(M_1+M_2)\geq 0$. Consider
now the two different measurements $\{M_1,M_2,M_3\}$ and
$\{M_1+M_2,M_3\}$. The second measurement is simply a
coarse-grained version of the first in which the two first
outcomes are grouped together. A direct application of property 2
above implies that $\Lambda(M_1)+\Lambda(M_2)=\Lambda(M_1+M_2)$.
This together with properties 1 and 2 imply that the map
$\Lambda$, initially defined for positive operators, can be
uniquely extended to a linear map acting on all operators. It
immediately follows that it can be written as $\tr (X M)$ for an
operator $X$. But then, the condition 1 implies the positivity of
the operator $X$ and its normalization follows from condition 2.
On the other hand, one checks by hand that any of these maps
satisfies conditions 1 and 2.

This theorem is a seminal result in the Foundations of Quantum Physics.
In particular, it implies that Born's rule for the computation of
measurement probabilities can be derived from the Hilbert space
structure of quantum measurements and the two natural conditions
provided above.

\subsection{Gleason correlations}

Gleason's Theorem was initially established for single systems. It
was later extended to composite systems in
Refs.~\cite{localgleason,wallach}. The scenario consists of $N$
independent observers. To each observer $j$, with $j=1,\ldots,N$,
one associates a Hilbert space of dimension $d_j$ and a structure
of quantum measurements given by sets of positive operators
summing up to the identity. For the sake of simplicity, we take in
what follows all the local dimensions equal, $d_i=d,\,\forall i$.
We denote by $\{M_{i_j}^{(j)}\},\,i_j=1,\ldots,k_j$ the sets of
positive operators defining a measurement for each observer, that
is, $\sum_{i_j} M_{i_j}^{(j)}=\one$. The extension of the
theorem then aims at characterizing those maps from measurements
by each observer to probability distributions. In what follows,
for the ease of notation, we restrict the analysis to the
simplest bipartite case, although it can be easily generalized
to an arbitrary number of parties. The map is requested to satisfy
the following conditions:

\begin{enumerate}
    \item For pairs of positive operators,
    $M_{i_1}^{(1)},M_{i_2}^{(2)}$, where $0\leq M_{i_1}^{(1)},M_{i_2}^{(2)}\leq\one$ one has
    $\Lambda(M_{i_1}^{(1)},M_{i_2}^{(2)})\geq 0$.
    \item For pairs of measurements,
    $\{M_{i_1}^{(1)}\},\{M_{i_2}^{(2)}\}$, where
    $0\leq M_{i_1}^{(1)},M_{i_2}^{(2)}\leq\one$ one has
    \begin{equation}\label{glcond}
    \sum_{i_1,i_2=1}^{k_1,k_2}\Lambda(M_{i_1}^{(1)},M_{i_2}^{(2)})=1 .
    \end{equation}
    \item Given two complete quantum measurements by one of the observers, say the second,
    $\{M_{i_2}^{(2)}\}$ and $\{N_{i_2}^{(2)}\}$, the map has to be such
    that
    \begin{equation}\label{lglcond}
    \sum_{i_2=1}^{k_2} \Lambda(M_{i_1}^{(1)},M_{i_2}^{(2)})=
    \sum_{i_2=1}^{k'_2} \Lambda(M_{i_1}^{(1)},N_{i_2}^{(2)}) .
    \end{equation}
\end{enumerate}

The new condition, i.e., the third one, can be understood as the natural formalization
of the no-signalling principle in the considered framework: the
marginal probability distribution seen by one of the observers
cannot depend on the measurement performed by the other observer.
The generalization to an arbitrary number of parties of these
requirements is straightforward. Now $\Lambda$ maps tuple of
positive operators $M_{i_1}^{(1)},\ldots,M_{i_N}^{(N)}$ into
non-signalling probability distributions.

The generalization of the theorem to this scenario, that we call
multipartite Gleason's Theorem, states that all such maps can be
written as
\begin{equation}\label{mgltheorem}
    \Lambda(M_{i_1}^{(1)},\ldots,M_{i_N}^{(N)})=
    \tr\left(W M_{i_1}^{(1)}\otimes\ldots\otimes M_{i_N}^{(N)}\right)
    ,
\end{equation}
where $W$ is an operator which is positive on product states
$\ket{\psi_1}\ldots\ket{\psi_N}$. These operators are also known
as entanglement witnesses \cite{PHTerhal}p. As above, it is clear that maps of the
form~\eqref{mgltheorem} satisfy the previous three requirements
and the non-trivial part of the result is proving the opposite
direction.

As the set of entanglement witnesses is larger than the set of
quantum states (or, in other words, there exist operators $W$ that
are non-positive, but positive on product states) the set of
distributions~\eqref{mgltheorem}, called in what follows Gleason
correlations, is in principle larger than the quantum set.
However, it was shown in~\cite{Barnum,Universal} that the two sets
actually coincide for two parties. Thus, as it happens for
single-party systems, imposing the structure of quantum
measurements for the observers gives the quantum correlations.

The proof of the equivalence between Gleason and bipartite quantum
correlations exploits the Choi-Jamio\l{}kowski (CJ) isomorphism
\cite{JC} that relates maps to operators. In this case, the it
says that any witness $W$ can be written as
$(I\otimes\Upsilon)(\Phi)$, where $\Upsilon$ is a positive map and
$\Phi$ is the projector onto the maximally entangled state
$\ket{\Phi}=(1/\sqrt{d})\sum_{i}\ket{ii}\in\mathbbm{C}^d\ot\mathbbm{C}^d$
and $I$ stands for an identity map. With the aid of
Ref.~\cite{HHHContr}, one can prove that any normalized witness
can also be written as $(I\otimes\Lambda)(\Psi)$, where $\Lambda$
is now a positive and trace-preserving map, while $\Psi$ is a
projector onto some pure bipartite state\footnote{To see this
explicitly let us first notice that for a normalized witness $W$
it holds that $W=(I\ot \Lambda)(\Phi)$ with trace-preserving
$\Lambda$ iff $W_A=\tr_B W=\mathbbm{1}/d$. Then, if $W_A\neq
\mathbbm{1}/d$ but it is of full rank, one introduces another
witness $\widetilde{W}=(1/d)(W_A^{-1/2}\ot \one)W(W_A^{-1/2}\ot
\one)$. Clearly, $\widetilde{W}_A=\mathbbm{1}/d$ and thus
$\widetilde{W}$ is isomorphic to a trace-preserving positive map
$\widetilde{\Lambda}$. Consequently,
\begin{equation}\label{foot}
W=d(\sqrt{W_A}\ot\one)\widetilde{W}(\sqrt{W_A}\ot\one)=d(\sqrt{W_A}\ot\one)(I\ot\widetilde{\Lambda})(\Phi)(\sqrt{W_A}\ot\one)=
(I\ot\widetilde{\Lambda})(\Psi),
\end{equation}
where $\Psi$ denotes a projector onto some normalized pure state
$\ket{\Psi}=\sqrt{d}(\sqrt{W_A}\ot\one)\ket{\Phi}$ of full Schmidt
rank. Finally, if $W_A$ is rank-deficient, one constructs yet
another witness $W'=W+\mathcal{P}_A^{\perp}\ot\one$, where
$\mathcal{P}_A^{\perp}=\one-\mathcal{P}_A$ with $\mathcal{P}_A$
denoting a projector onto the support of $W_A$. Then, $W'_A$ is of
full-rank and therefore $W'$ admits the form (\ref{foot}). To
complete the proof, it suffices to notice that
$W=(\mathcal{P}_A\ot\one)W'(\mathcal{P}_A\ot\one)$, and hence $W$
also assumes the form (\ref{foot}) with a normalized pure state
$\ket{\Psi}=\sqrt{d}[\mathcal{P}_A(W_A')^{1/2}\ot\one]\ket{\Phi}=\sqrt{d}(W_A^{1/2}\ot\one)\ket{\Phi}$
which is now not of full Schmidt rank.}. It then follows that
\begin{eqnarray}\label{bipgleason}
    \tr(W M_{a_1}^{x_1}\otimes M_{a_2}^{x_2})&=&
    \tr[(I\otimes\Lambda)(\Psi) M_{a_1}^{x_1}\otimes
    M_{a_2}^{x_2}]\nonumber\\
    &=&\tr[\Psi
    M_{a_1}^{x_1}\otimes\Lambda^*(M_{a_2}^{x_2})]\nonumber\\
    &=&
    \tr(\Psi M_{a_1}^{x_1}\otimes\widetilde M_{a_2}^{x_2}) ,
\end{eqnarray}
where $\Lambda^*$ is the dual\footnote{The dual map $\Lambda^*$ of
$\Lambda$ is the map such that
$\tr(A\Lambda(B))=\tr[\Lambda^*(A)B]$.} of $\Lambda$ and
$\widetilde M_{a_2}^{x_2}=\Lambda^*(M_{a_2}^{x_2})$ defines a
valid quantum measurement because the dual of a positive
trace-preserving map is positive and unital, that is,
$\Lambda^*(\one)=\one$.

The next natural question is as to whether the equivalence between
quantum and Gleason correlations holds for an arbitrary number of
parties. As we show next, the answer to this question turns out to
be negative, as there are local measurements acting on
entanglement witnesses that produce supra-quantum correlations.
Before proving this result, it is worth mentioning that local
measurements on entanglement witnesses that can be written as
\begin{equation}
    W=\sum_k \big(\Lambda^k_{A_1}\otimes\cdots\otimes
    \Lambda^k_{A_N}\big)(\rho_k) ,
\end{equation}
where $\rho_k$ are $N$-party quantum states, $\Lambda^k_{A_i}$ are
positive trace preserving maps and the number of terms in the sum
is arbitrary, do not lead to supra-quantum correlations. This is a
rather straightforward generalization of the equivalence proof in
the bipartite case.

In order to prove that in the multipartite case the set of Gleason
correlations contains quantum correlations as a strict subset, we
provide an example of entanglement witness and local measurements
giving nonsignalling correlations which violate the three-partite
GYNI Bell inequality. Let us start by introducing the following
set of four fully product vectors from the three-qubit Hilbert
space:
\begin{eqnarray}
\label{upb3q}
  \ket{\psi_1}= \ket{000},\qquad  \ket{\psi_2}= \ket{1e^\bot e},\qquad \ket{\psi_3} = \ket{e1e^\bot},\qquad
  \ket{\psi_4} = \ket{e^\bot e1},
\end{eqnarray}
where $\ket{e}\in\mathbbm{C}^2$ is an arbitrary vector different
from $\ket{0}$ and $\ket{1}$, while $\ket{\overline{e}}$ stands
for a vector orthogonal to $\ket{e}$. One checks by hand that
there is no other three-qubit fully product vector orthogonal to
all $\ket{\psi_i}$s; such sets of product vectors are called
\textit{unextendible product bases} (UPBs) \cite{BennettUPB} (see
section \ref{UPB} for a detailed discussion on UPBs and more
examples).

As noticed in \cite{BennettUPB}, the set (\ref{upb3q}), called
Shifts UPB, can be used for a simple construction of bound
entangled state, i.e., an entangled state from which any type of
maximally entangled state cannot be distilled
\cite{HorodeckiBound}. The state is given by
$\rho_\mathrm{UPB}=(\one-\Pi_\mathrm{UPB})/4$ with
$\Pi_\mathrm{UPB}$, where $\Pi_\mathrm{UPB}$ denotes the projector
onto $\mathrm{span}\{\ket{\psi_i}\}$.

Let us now consider the normalized entanglement witness detecting
$\rho$:
\begin{equation}\label{witness}
    W=\frac{1}{4-8\epsilon}(\Pi_\mathrm{UPB}-\epsilon\mathbbm{1}),
\end{equation}
where
\begin{equation}
\epsilon=\min_{\ket{\alpha \beta \gamma}}\bra{\alpha
\beta \gamma}\Pi_\mathrm{UPB}\ket{\alpha \beta \gamma}.
\end{equation}
The fact that there is no fully product vector orthogonal to
$\ket{\psi_i}$ implies that $\epsilon>0$, and, on the other hand,
it is fairly easy show that $\epsilon<1/2$. One also notices that
$\tr(W\rho)=-\epsilon/(1-2\epsilon)<0$.

Now, one can see that the witness $W$, when measured along the
local bases in the definition of the UPB~\eqref{upb3q}, leads to
correlations that produce a value of GYNI game equal to
$\beta=(1-\epsilon)/(1-2\epsilon)$, which is larger than one for
all positive $\epsilon$ not larger than one-half. Thus, these
correlations represent an example of Gleason correlations with no
quantum analogue.

\section{Application II: Quantum correlations and information principles}
\label{sec:qcorr}

As mentioned in the introduction, an intense research effort has
recently been devoted to understand why nonlocality appears to be
limited in quantum mechanics. Information concepts have been
advocated as the key missing ingredient needed to single-out the
set of quantum correlations \cite{brassardNP,popescuNP,bub}. The
main idea is to identify `natural' information principles,
satisfied by quantum correlations, but violated by super-quantum
correlations. The existence of the latter would then have
implausible consequences from an information-theoretic point of
view. Celebrated examples of these principles are information
causality \cite{Pawlowski2009a} or non-trivial communication
complexity \cite{Dam2005}. While the use of these information
concepts has been successfully applied to specific scenarios
\cite{Brassard2006,Brunner2009,Allcock2009,Ahanj2010,Cavalcanti2010},
proving, or disproving, the validity of a principle for quantum
correlations is extremely challenging. On the one hand, it is
rather difficult to derive the Hilbert space structure needed for
quantum correlations from information quantities. On the other
hand, proving that some super-quantum correlations are fully
compatible with an information principle seems out of reach, as
one needs to consider all possible protocols using these
correlations and show that none of them leads to a violation of
the principle. Hence it is still unclear whether this approach is
able to fully recover the set of quantum correlations.

Therefore it is relevant to derive general features of a principle that
could potentially identify quantum correlations. Using GYNI, it was recently
shown that such a principle must be genuinely multipartite. More specifically,
no bipartite principle can characterize the set of quantum correlations when
three of more observers are involved~\cite{gwan}. This rest of this section
is devoted to this result.

Before discussing the result, it is worth
recalling that, so far, most information-theoretic
principles have been formulated in the bipartite scenario.
Actually, even the general formulation of the no-signalling principle has a bipartite
structure: correlations among $N$ observers are
compatible with the no-signalling principle whenever there exists
no partition of the $N$ parties into two groups such that the
marginal probability distribution of one set of the parties
depends on the measurements performed by the other set of parties (see \eqref{mnscorr}).
Moving to information causality, it considers a scenario in which
a first party, Alice, has a string of $n_A$ bits. Alice is then
allowed to send $m$ classical bits to a second party, Bob.
Information causality bounds the information Bob can gain on the
$n_A$ bits held by Alice whichever protocol they implement making
use of the pre-established bipartite correlations and the message
of $m$ bits. Alice and Bob can violate this principle when they
have access to some super-quantum correlations
\cite{Pawlowski2009a}. In the case $m=0$, information causality
implies that in absence of a message, pre-established correlations
do not allow Bob to gain any information about any of the bits
held by Alice, which is nothing but the no-signaling principle.
This suggests the following generalization of information
causality to an arbitrary number of parties, mimicking what is
done for the no-signalling principle: given some correlations
$P(a_1,\ldots,a_n|x_1,\ldots,x_N)$, they are said to be compatible
with information causality whenever all bipartite correlations
constructed from them satisfy this principle. This generalization
ensures the correspondence between no-signaling and information
causality when $m=0$ for an arbitrary number of parties. This
generalization of information causality has recently been applied
to the study of extremal tripartite non-signaling
correlations~\cite{singapore}.

Regarding non trivial communication complexity, it studies how
much communication is needed between two distant parties to
compute probabilistically a function of some inputs in a
distributed manner. It can also be interpreted as a generalization
of the no-signaling principle, as it imposes constraints on
correlations when a finite amount of communication is allowed
between parties. Different multipartite generalizations of the
principle have been studied, see \cite{CC}. However, as
for information causality, one can always consider the
straightforward generalization in which the principle is applied
to every partition of the $N$ parties in two groups.

We are now in position to review the proof
of the impossibility of characterizing quantum correlations for an
arbitrary number of parties using bipartite principles. For
simplicity, we restrict the analysis to tripartite correlations.

\subsection{Time-ordered-bilocal correlations and GYNI}

The first ingredient in the proof is the characterization of
multipartite correlations such that any bipartite correlations
constructed from them have a classical local model. By definition,
correlations satisfying this property do not violate any bipartite
principle satisfied by classical correlations.

A priori, one would think that if the correlations
$P(a_1,a_2,a_3|x_1,x_2,x_3)$ have a local model along all possible
bipartitions, namely $A_1-A_2A_3$, $A_2-A_1A_3$ and $A_3-A_1A_2$,
that is,
\begin{eqnarray}\label{bilocal}
    P(a_1,a_2,a_3|x_1,x_2,x_3)&=&\sum_\lambda P_1(\lambda)P_1(a_1|x_1,\lambda)P_1(a_2,a_3|x_2,x_3,\lambda)\nonumber\\
    &=&\sum_\lambda P_2(\lambda)P_2(a_2|x_2,\lambda)P_2(a_1,a_3|x_1,x_3,\lambda)\nonumber\\
    &=&\sum_\lambda P_3(\lambda)P_3(a_3|x_3,\lambda)P_1(a_1,a_2|x_1,x_2,\lambda) ,
\end{eqnarray}
then, any bipartite object constructed from it also has a local
model. This intuition however has proven to be wrong
in~\cite{gwan2}, where it was shown how non-local bipartite
correlations can be derived from correlations having a
decomposition of the form of~\eqref{bilocal}. The characterization
of multipartite correlations such that a local model exists for
any bipartite correlations derived from it is then subtler than
expected. Indeed, at the moment, it is unknown what is the largest
set of correlations having this property~\cite{gwan2}. It has
however been shown in~\cite{gwan} that the set of
time-ordered-bilocal correlations (TOBL) do fulfill this
requirement. Tripartite correlations have a TOBL model whenever
they can be written as
\begin{eqnarray}
\label{eq:tobl}
  P(a_1, a_2, a_3|x_1, x_2, x_3)
  &=& \sum_\lambda P_\lambda^{i|jk} P(a_i|x_i, \lambda) P_{j\rightarrow k}(a_j, a_k|x_j,
  x_k,\lambda)\nonumber\\
  &=& \sum_\lambda P_\lambda^{i|jk} P(a_i|x_i, \lambda) P_{j\leftarrow k}(a_j, a_k|x_j,
  x_k,\lambda)
\end{eqnarray}
for $(i,j,k)= (1,2,3), (2,3,1),(3,1,2)$, with the distributions
$P_{j\rightarrow k}$ and $P_{j\leftarrow k}$ obeying the
conditions
\begin{align}
 \label{eq:timeordered1}
  &P_{j \rightarrow k}(a_j|x_j,\lambda) = \sum_{a_k} P_{j \rightarrow k}(a_j, a_k|x_j, x_k,\lambda),\\
  \label{eq:timeordered2}
  &P_{j \leftarrow k}(a_k|x_k,\lambda) = \sum_{a_j} P_{j \leftarrow k}(a_j, a_k|x_j, x_k,\lambda).
\end{align}
The notion of TOBL correlations first appeared
in~\cite{Pironio2011} (see~\cite{gwan2} and~\cite{tobl2} for a
proper introduction and further motivation for such models). As
can be seen from the relations \eqref{eq:timeordered1} and
\eqref{eq:timeordered2} we impose the distributions
$P_{j\rightarrow k}$ and $P_{j\leftarrow k}$ to allow for
signaling at most in one direction, indicated by
the arrow (see Table \ref{tab:point}). 

\begin{table}[ht]

 \centering
 \begin{center}
\begin{tabular}{c|c|c|c}
$x_2$ & $x_3$ & $a_2$ & $a_3$ \\
\hline \hline
$0$ & $0$ & $0$ & $0$\\
\hline
$0$ & $1$ & $0$ & $1$\\
\hline
$1$ & $0$ & $1$ & $1$\\
\hline
$1$ & $1$ & $1$ & $0$\\
 \end{tabular}
\qquad
\begin{tabular}{c|c|c|c}
$x_2$ & $x_3$ & $a_2$ & $a_3$ \\
\hline \hline
$0$ & $0$ & $0$ & $0$\\
\hline
$0$ & $1$ & $1$ & $1$\\
\hline
$1$ & $0$ & $0$ & $0$\\
\hline
$1$ & $1$ & $1$ & $1$\\
 \end{tabular}
\qquad
\begin{tabular}{c|c|c|c}
$x_2$ & $x_3$ & $a_2$ & $a_3$ \\
\hline \hline
$0$ & $0$ & $1$ & $0$\\
\hline
$0$ & $1$ & $0$ & $1$\\
\hline
$1$ & $0$ & $0$ & $1$\\
\hline
$1$ & $1$ & $1$ & $0$\\
 \end{tabular}
 \end{center}
\caption{Different examples of deterministic bipartite probability
distributions $P_{23}(a_2, a_3|x_2,x_3,\lambda)$ characterized by
output assignments to the four possible combination of
measurements. Left: inputs and outputs corresponding to a point
$P_{2 \rightarrow 3}(a_2, a_3|x_2,x_3,\lambda)$ in the decomposition
\eqref{eq:tobl}. Center: inputs and outputs corresponding to a
point $P_{2 \leftarrow 3}(a_2, a_3|x_2,x_3,\lambda)$ in
\eqref{eq:tobl}. Right: inputs and outputs corresponding to a
distribution which allows signaling in the two directions.}
 \label{tab:point}
\end{table}

To understand the operational meaning of these models, consider
the bipartition $1|23$ for which systems $2$ and $3$ act together.
In this situation, $P(a_1,a_2,a_3|x_1,x_2,x_3)$ can be simulated if a
classical random variable $\lambda$ with probability distribution
$p^{1|23}_\lambda$ is shared by parts $1$ and the composite system
$2-3$, and they implement the following protocol: given $\lambda$,
$1$ generates its output according to the distribution $P(a_1|x_1,
\lambda)$; on the other side, and depending on which of the
parties $2$ and $3$ measures first, $2-3$ uses either $P_{2
\rightarrow 3}(a_2, a_3|x_2, x_3,\lambda)$ or $P_{2 \leftarrow
3}(a_2, a_3|x_2, x_3,\lambda)$ to produce the two measurement
outcomes. Likewise, any other bipartition of the three systems
admits a classical simulation.

By construction, the set of tripartite TOBL models is convex and
is included (in fact, it is strictly included~\cite{gwan2}) in the
set of tripartite probability distributions of the form
(\ref{bilocal}). Moreover, TOBL models always produce classical
correlations under post-selection: indeed, suppose that we are
given a tripartite distribution $P(a_1,a_2,a_3|x_1,x_2,x_3)$
satisfying condition (\ref{eq:tobl}), and a postselection is made
on the outcome $\tilde{a}_3$ of measurement $\tilde{x}_3$ by party
$3$. Then, one has
\begin{equation}
  \hspace{-10pt}P(a_1,a_2|x_1,x_2\tilde{x}_3\tilde{a}_3)
  =\sum_\lambda P'_\lambda P(a_1|x_1,\lambda)P'(a_2|x_2, \lambda),
  \label{post_sel}
\end{equation}

\noindent with

\begin{equation}
P'_\lambda=\frac{P_\lambda^{1|23}}{P(\tilde{a}_3|\tilde{x}_3)}P_{2\leftarrow
3}(\tilde{a}_3|\tilde{x}_3,\lambda), \qquad
P'(a_2|x_2,\lambda)=P_{2\leftarrow
3}(a_2|x_2,\tilde{x}_3,\tilde{a}_3, \lambda).
\end{equation}
Postselected tripartite TOBL boxes can thus be regarded
as elements of the TOBL set with trivial outcomes for one of the
parties.

We now demonstrate that any possible bipartite correlations
derived from many uses of TOBL correlations have a local model
and, thus, are compatible with any bipartite principle satisfied
by classical (and obviously quantum) correlations. The most
general protocol consists in distributing an arbitrary number of
boxes described by $P^1, P^2,\ldots, P^N$ among three parties
which are split into two groups, $A$ and $B$. Both groups can
process the classical information provided by their share of the
$N$ boxes. For instance, outputs generated by some of the boxes
can be used as inputs for other boxes (see
figure~\ref{fig:wirings}). This local processing of classical
information is usually referred to as \textit{wirings}. Thus, in
order to prove our result in full generality, we should consider
all possible wirings of tripartite boxes. We show next that if
$P^1, P^2,\ldots, P^N$ are in TOBL, then the resulting
correlations $P_{\mathrm{fin}}$ obtained after any wiring protocol
have a local decomposition with respect to the bipartition $A|B$,
and therefore fulfill any bipartite information principle.

For simplicity, we illustrate our procedure for the wiring shown
in figure~\ref{fig:wirings}, where boxes $P^{1}, P^{2} , P^{3}$
are distributed between two parties $A$ and $B$, and party $A$
only holds one subsystem of each box. The construction is
nevertheless general: it applies to any wiring and also covers
situations where for some TOBL boxes party $A$ holds two
subsystems instead of just one (or even the whole box).

From (\ref{eq:tobl}) we have
\begin{eqnarray}
\label{eq:tomodels}
  P^{i}(a^{i}_1,a^{i}_2,a^{i}_3| x^{i}_1,x^{i}_2,x^{i}_3)&=& \sum_{\lambda^i}
  P_{\lambda^i}^{i} P_1^{i}(a^{i}_1|x_1^{i}, \lambda^i)P_{2 \rightarrow 3}^{i}(a_2^{i},a_3^{i}|x^{i}_2,x^{i}_3, \lambda^i)\\
&=& \sum_{\lambda^i} P_{\lambda^i}^{i} P_1^{i}(a^{i}_1|x^{i}_1,
\lambda^i)P_{2 \leftarrow 3}^{i}(a^{i}_2,a_3^{i}|x_2^{i},x_3^{i},
\lambda^i),
\end{eqnarray}
\noindent for $i= 1,2, 3$. Consider the first box that receives an
input, in our case subsystem $2$ of $P^{1}$. The first outcome
$a_2^{1}$ can be generated by the probability distribution $P_{2
\rightarrow 3}^{1}(a_2^{1},a_3^{1}|x_2^{1},x_3^{1},\lambda^{1})$
encoded in the hidden variable $\lambda^1$ that models these first
correlations. This is possible because for this decomposition
$a_2^{1}$ is defined independently of $x_3^{1}$, the input in
subsystem $3$. Then, the next input $x_3^{2}$, which is equal to
$a_2^{1}$, generates the output $a_3^{2}$ according to the
probability distribution $P_{2 \leftarrow
3}^{2}(a_2^{2},a_3^{2}|x_2^{2},x_3^{2},\lambda^{2})$ encoded in
$\lambda^2$. The subsequent outcomes $a_2^{i}$ and $a_3^{i}$ are
generated in a similar way. The general idea is that outputs are
generated sequentially using the local models according to the
structure of the wiring on $2-3$. Finally, subsystem $1$ can
generate its outputs $a^{i}$ by using the probability distribution
$P_{1}^{i}(a_1^{i}|x^{i},\lambda^{i})$. This probability
distribution is independent of the order in which parties $2$ and
$3$ make their measurement choices for any of the boxes. Averaging
over all hidden variables one obtains $P_\mathrm{fin}$. This
construction provides the desired local model for the final
probability distribution.

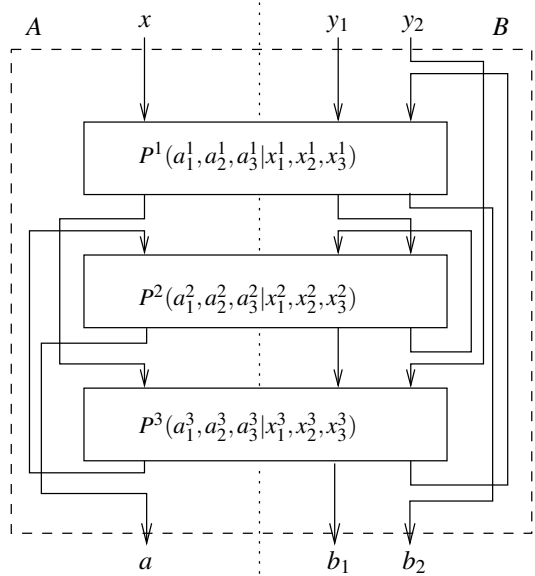
\begin{figure}
\begin{center}
\begin{picture}(0,0)%
\includegraphics{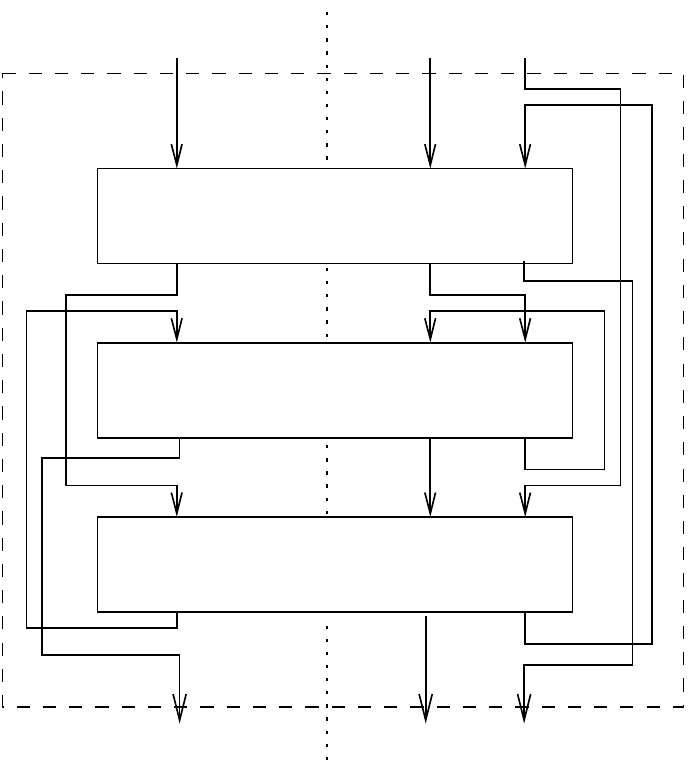}%
\end{picture}%
\setlength{\unitlength}{4144sp}%
\begingroup\makeatletter\ifx\SetFigFont\undefined%
\gdef\SetFigFont#1#2#3#4#5{%
  \reset@font\fontsize{#1}{#2pt}%
  \fontfamily{#3}\fontseries{#4}\fontshape{#5}%
  \selectfont}%
\fi\endgroup%
\begin{picture}(3138,3496)(4309,505)
\put(5086,2999){\makebox(0,0)[lb]{\smash{{\SetFigFont{9}{10.8}{\familydefault}{\mddefault}{\updefault}{\color[rgb]{0,0,0}$P^{1}(a_1^1,a_2^1,a_3^1|x_1^1,x_2^1,x_3^1)$}%
}}}}
\put(6661,3809){\makebox(0,0)[lb]{\smash{{\SetFigFont{10}{12.0}{\familydefault}{\mddefault}{\updefault}{\color[rgb]{0,0,0}$y_2$}%
}}}}
\put(5086,569){\makebox(0,0)[lb]{\smash{{\SetFigFont{10}{12.0}{\familydefault}{\mddefault}{\updefault}{\color[rgb]{0,0,0}$a$}%
}}}}
\put(5086,3809){\makebox(0,0)[lb]{\smash{{\SetFigFont{10}{12.0}{\familydefault}{\mddefault}{\updefault}{\color[rgb]{0,0,0}$x$ }%
}}}}
\put(6211,3809){\makebox(0,0)[lb]{\smash{{\SetFigFont{10}{12.0}{\familydefault}{\mddefault}{\updefault}{\color[rgb]{0,0,0}$y_1$}%
}}}}
\put(6211,569){\makebox(0,0)[lb]{\smash{{\SetFigFont{10}{12.0}{\familydefault}{\mddefault}{\updefault}{\color[rgb]{0,0,0}$b_1$}%
}}}}
\put(6661,569){\makebox(0,0)[lb]{\smash{{\SetFigFont{10}{12.0}{\familydefault}{\mddefault}{\updefault}{\color[rgb]{0,0,0}$b_2$}%
}}}}
\put(4411,3764){\makebox(0,0)[lb]{\smash{{\SetFigFont{10}{12.0}{\rmdefault}{\mddefault}{\updefault}{\color[rgb]{0,0,0}$A$}%
}}}}
\put(7201,3764){\makebox(0,0)[lb]{\smash{{\SetFigFont{10}{12.0}{\rmdefault}{\mddefault}{\updefault}{\color[rgb]{0,0,0}$B$}%
}}}}
\put(5086,1379){\makebox(0,0)[lb]{\smash{{\SetFigFont{9}{10.8}{\familydefault}{\mddefault}{\updefault}{\color[rgb]{0,0,0}$P^{3}(a_1^3,a_2^3,a_3^3|x_1^3,x_2^3,x_3^3)$}%
}}}}
\put(5086,2144){\makebox(0,0)[lb]{\smash{{\SetFigFont{9}{10.8}{\familydefault}{\mddefault}{\updefault}{\color[rgb]{0,0,0}$P^{2}(a_1^2,a_2^2,a_3^2|x_1^2,x_2^2,x_3^2)$}%
}}}}
\end{picture}%
\end{center}
\caption{\label{fig:wirings}Wiring of several tripartite
correlations distributed among parties $A$ and $B$. The generated
bipartite box accepts a bit $x$ (two bits $y_1,y_2$) as input on
subsystem $A$ ($B$) and returns a bit $a$ (two bits $b_1,b_2$) as
output. Relations \eqref{eq:tomodels} guarantee that the final
bipartite distribution $P_\mathrm{fin}(a,(b_1,b_2)|x,(y_1,y_2))$
admits a local model.}
\end{figure}

The final step in the proof consists of showing that there exist
correlations in the TOBL set that do not have a quantum
realization. This was shown by means of the GYNI inequality. More
precisely, it can be proven that, contrary to quantum
correlations, this inequality is violated by TOBL correlations:
\begin{equation}\label{eq:linprog}
\begin{aligned}
&{\text{maximize}}\;\; P(000|000)+P(110|011)+P(011|101)+P(101|110)
 \\
&\text{subject to}\;\; P(a_1, a_2, a_3|x_1, x_2, x_3)\in\mathrm{TOBL}.
\end{aligned}
\end{equation}
The maximization yields a value of $7/6$, implying the existence
of supra-quantum correlations in TOBL. The form of the TOBL
correlations leading to this violation can be found
in~\cite{gwan}. Later, another example of supra-quantum
correlations in TOBL was provided in~\cite{singapore}, where the
authors proved that an extremal point of the no-signalling
polytope for three parties and two two-outcome measurements per
party is also in TOBL and has no quantum realization.

\section{Generalization of GYNI: Bell inequalities without quantum violation and unextendible product bases}
\label{sec:upb}

The relation between GYNI's Bell inequality and the three-qubit
unextendible product basis (UPB) was used in the previous section
to show that, contrary to the bipartite case \cite{Barnum} (see
also \cite{Universal}), in the three-partite scenario Gleason
correlations make a larger set than the quantum ones. Actually,
this link can be generalized and relates nontrivial Bell
inequalities without violation to UPB, see Ref. \cite{BellUPB1}.
All these Bell inequalities lack quantum violation, nevertheless,
they are nontrivial in the sense that there exist some
nonsignalling correlations violating them. They therefore
complement the results of Ref. \cite{Universal} providing new
examples of multipartite scenarios where Gleason correlations are
different from the quantum ones. More importantly, however, some
of UPBs can lead to {\it tight} Bell inequalities with no quantum
violation, novel examples of which have recently been found
\cite{BellUPB2}.

Our aim in this section is to recall the method from Refs.
\cite{BellUPB1,BellUPB2} and then discuss properties of the
resulting Bell inequalities. We also provide some classes of
nontrivial Bell inequalities with no quantum violation associated
to UPBs. Finally, we go beyond UPB and show that there are also
sets of orthogonal product vectors that are not UPBs but can be
associated nontrivial Bell inequalities. Before that let us recall
the notion of unextendible product bases and briefly review their
properties.

\subsection{Unextendible product bases}
\label{UPB}

We start by introducing an $N$-partite product Hilbert space
\begin{equation}\label{HilbertSpace}
H=\mathbbm{C}^{d_1}\ot\ldots\ot\mathbbm{C}^{d_N},
\end{equation}
where $d_i$ $(i=1,\ldots,N)$ denote, for the time being arbitrary,
dimensions of the local Hilbert spaces. In what follows we will
call an element $\ket{\psi}$ of $H$ {\it fully product} if it
assumes the form
$\ket{\psi}=\otimes_{i=1}^{N}\ket{\psi_i}\equiv\ket{\psi_1,\ldots,\psi_N}$
with $\ket{\psi_i}\in\mathbbm{C}^{d_i}$.

Then, let us consider a set of orthogonal product vectors
\begin{equation}\label{S}
S=\left\{\ket{\Psi_m}=\ket{\psi_{m}^{(1)}}\ot\ldots\ot\ket{\psi_m^{(N)}}\right\}_{m=1}^{|S|},
\end{equation}
where $\ket{\psi_{m}^{(i)}}$ $(m=1,\ldots,|S|)$ are local vectors
belonging to $\mathbbm{C}^{d_i}$ and $|S|\leq \mathrm{dim}H$. With
this we have the following definition \cite{BennettUPB}.

\begin{definition}\label{def:upb}
Let $S$ be a set of orthogonal fully product vectors (\ref{S})
from $H$. We call $S$ {\it unextendible product basis} (UPB) if it
spans a proper subspace in $H$, i.e., $|S|<\dim H$, and there is
no product vector $\ot_{i=1}^{N}\ket{\phi_i}\in H$ orthogonal to
$\mathrm{span}S$.
\end{definition}

\begin{figure}[]
\sidecaption
\includegraphics[width=0.4\textwidth]{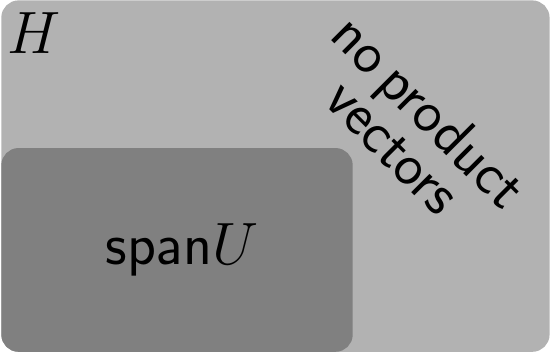}
\caption{Schematic definition of a UPB: a set of orthogonal
product vectors $S$ spanning a proper subspace
$\mathrm{span}S\subset H$ such that there is no vector
$\ot_{i=1}^{N}\ket{\phi_i}\in H$ orthogonal to $S$. A normalized
projector onto $(\mathrm{span}S)^{\perp}$ (\ref{BE}) is a bound
entangled state \cite{BennettUPB}.}\label{figura:BE}
\end{figure}

The notion of unextendible product bases reflects the peculiar
feature of some of product Hilbert spaces $H$ in that they can be
represented as direct sums of two orthogonal subspaces, of which
one is spanned by product vectors, while the second does not
contain any of them, i.e., is completely entangled (see Fig.
\ref{figura:BE}). This has interesting consequences from the
quantum information point of view. As it was first observed by
Bennett and coworkers \cite{BennettUPB}, UPBs can be used for a
construction of bound entangled states, i.e., states that are
entangled but nevertheless no entanglement can be distilled from
them by means of local operations and classical communication
\cite{HorodeckiBound}.

To be more precise, following \cite{BennettUPB}, let us consider a
particular UPB $U$, and the normalized projector onto the subspace
of $H$ orthogonal to $U$, i.e.,
\begin{equation}\label{BE}
\varrho=\frac{1}{\dim H-|U|}\left(\mathbbm{1}-\Pi\right).
\end{equation}
By $\Pi$ and $\mathbbm{1}$ we denoted, respectively, the projector onto the
subspace spanned by $U$ and identity acting on $H$.
Since there is not product vector orthogonal to $U$, the support
of $\varrho$ consists only of entangled states, implying that
$\varrho$ must be entangled. Then, it immediately follows from Eq.
(\ref{BE}) that $\varrho$ has all partial transpositions positive
which, due to Ref. \cite{HorodeckiBound}, justifies the statement that
$\varrho$ is bound entangled.

To illustrate the above definition we consider the following examples of
UPBs.

\begin{example}We start from the TILES UPB, one of the first
bipartite UPBs introduced in Ref. \cite{BennettUPB}. It consists
of five two-qutrit vectors of the form
\begin{eqnarray}\label{TILES}
U_{\mathrm{TILES}}&=&\{\ket{0}(\ket{0}-\ket{1}), \ket{2}(\ket{1}-\ket{2}),
(\ket{0}-\ket{1})\ket{2}, (\ket{1}-\ket{2})\ket{0},\nonumber\\
&&\hspace{0.5cm}(\ket{0}+\ket{1}+\ket{2})^{\ot 2}\}.
\end{eqnarray}
Notice that in two-qutrit Hilbert space there only exist
five-elements UPBs and all of them are known
\cite{BennettUPB,UPBhuge,Leinaas,Skowronek}.
\end{example}

\begin{example}
Second, let us consider a general class of $N$-qubit unextendible
product bases with odd $N=2k-1$ $(k\in\mathbbm{N};k\geq 2)$ given
by the following $2k$ vectors \cite{UPBhuge}:
\begin{eqnarray}\label{GenShifts}
U_{\mathrm{GenShifts}}&=&\{\ket{0\ldots 0},\ket{1e_1\ldots e_{k-1}\overline{e}_{k-1}\ldots\overline{e}_1},
\ket{\overline{e}_11e_1\ldots e_{k-1}\overline{e}_{k-1}\ldots\overline{e}_2},
\ldots,\nonumber\\
&&\hspace{0.2cm}\ket{e_1\ldots e_{k-1}\overline{e}_{k-1}\ldots\overline{e}_11}\}
\end{eqnarray}
with $\{\ket{0},\ket{1}\}$ and
$\{\ket{e_i},\ket{\overline{e}_i}\}$ $(i=1,\ldots,k-1)$ being $k$
arbitrary but different bases in $\mathbbm{C}^2$. The $i$th
$(i\geq 2)$ vector in (\ref{GenShifts}), except for the first two
ones, is obtained from the vector $i-1$ by shifting all the local
vectors by one to the right, and thus the name {\it Generalized
Shifts}.
\end{example}

\begin{example}
Third, let us consider the general class of UPBs found by Niset
and Cerf \cite{NisetCerf}. Here we take the Hilbert space
$H=(\mathbbm{C}^{d})^{\ot N}$, where $N\geq 3$ and $d\geq N-1$,
and the following set of $N(d-1)+1$ vectors:
\begin{equation}\label{NC1}
U_{\mathrm{NC}}=\{\ket{e_{d-1}}^{\ot N}\}\cup \bigcup_{i=0}^{N-1}S_i,
\end{equation}
where
\begin{equation}\label{NC2}
S_0=\{\ket{0,1,\ldots,d-1}\ket{e_0},\ldots,\ket{0,1,\ldots,d-1}\ket{e_{d-2}}\}
\end{equation}
and $S_i=V^{i}S_0$ $(i=1,\ldots,N-1)$ with $V$ denoting a unitary
permutation operator such that
$V\ket{x_1}\ldots\ket{x_N}=\ket{x_N}\ket{x_{N-1}}\ldots\ket{x_1}$
for $\ket{x_i}\in\mathbbm{C}^d$, and $\{\ket{e_i}\}_{i=0}^{d-1}$
is any orthogonal basis in $\mathbbm{C}^d$ different from the
standard one. Notice that $U_{NC}$ can straightforwardly be
generalized to an arbitrary local dimension $d_i\geq N-1$
$(i=1,\ldots,N)$ just by adjusting both bases at each site to the
respective dimension \cite{NisetCerf}.
\end{example}

Both classes of multipartite UPBs from examples 2 and 3 (here up
to local unitary operations) recover, for $N=3$, the already
introduced Shifts UPB (\ref{upb3q}), i.e.,
$U_{\mathrm{Shifts}}=\{\ket{000},\ket{1\overline{e}e},\ket{e1\overline{e}},\ket{\overline{e}e1}\}$
with $\{\ket{e},\ket{\overline{e}}\}$ being an arbitrary basis of
$\mathbbm{C}^2$ different from the standard one. Clearly, this set
can be slightly generalized by taking the second basis different
at each site, that is,
$\{\ket{000},\ket{1\overline{e}_2e_3},\ket{e_11\overline{e}_3},\ket{\overline{e}_1e_21}\}$
(the first basis can be fix to the standard one by a local unitary
operation). Then, as it was shown by Bravyi \cite{Bravyi}, any
three-qubit UPB is equivalent to this one up to local unitary
operations and permutations of the parties.

\subsection{Constructing Bell inequalities with no quantum violation from unextendible product bases}

We are now ready to recall the method from
\cite{BellUPB1,BellUPB2} allowing to associate a nontrivial Bell
inequality with no quantum violation to a UPB having certain
property.

\subsubsection{The construction}
\label{sec:constr}

To begin, consider again the product Hilbert space $H$ and the set of
vectors $S$. For the time being, we do not assume $S$ to be a UPB,
keeping, however, the assumption that elements of $S$ are orthogonal
product vectors from $H$. Then, let us collect all different local vectors
appearing in all vectors $\ket{\Psi_m}$ at the $i$th site in the local
sets
\begin{equation}\label{Si}
S^{(i)}=\big\{\ket{\psi_m^{(i)}}\big\}_{m=1}^{s_i}\qquad
(i=1,\ldots,N),
\end{equation}
where $s_i\leq |S|$. Subsequently, among elements of $S^{(i)}$ we
search for mutually orthogonal vectors and collect them in
separate subsets $S_n^{(i)}$ $(n=0,\ldots,k_i)$ such that
$S_0^{(i)}\cup\ldots\cup S_{k_i}^{(i)}=S^{(i)}$ for any $i$ (see
Fig. \ref{fig:tree}). Notice that these subsets may, but do not
have to, span the corresponding Hilbert space $\mathbbm{C}^{d_i}$.

\begin{figure}[]
\sidecaption
\includegraphics[width=0.5\textwidth]{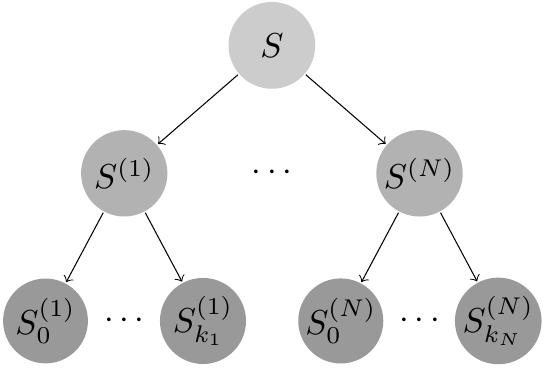}
\caption{Schematic description of our construction. From the set
$S$ having the local independence property, one constructs the
local sets $S^{(i)}$ $(i=1,\ldots,N)$ by collecting different
local vectors in $\ket{\Psi_m}$. Then, one distinguishes local
subsets $S^{(i)}_m$ of mutually orthogonal vectors among elements
of each local set $S^{(i)}$.}\label{fig:tree}
\end{figure}

It should be emphasized that there exist sets $S$ for which
the local subsets cannot be unambiguously defined. This is, for
instance, the case for vectors $U_{\mathrm{TILES}}$ [cf. Eq.
(\ref{TILES})]. At both sites there are five different vectors
$\ket{0}$, $\ket{0}-\ket{1}$, $\ket{1}-\ket{2}$, $\ket{2}$, and
$\ket{0}+\ket{1}+\ket{2}$. Clearly, the first one is orthogonal to
the third and fourth ones, however, the latter are not mutually
orthogonal. Then, in order to avoid this ambiguity, we consider
only those sets $S$ that have the following property.\\

\noindent{\bf Property.} {\it Let $S$ be a set of orthogonal
product vectors from $H$. No two vectors belonging to different
subsets $S^{(i)}_{k}$ and $S^{(i)}_{l}$ $(k\neq l)$, which are
constructed along the above lines, are orthogonal.\\}

In other words, what we need is that the local subsets are
constructed in the way that by replacing one of them by another
one of the same size, we keep the orthogonality of elements of
$S$. In yet another words, the above property guarantees that the
orthogonality of $S$ is preserved under any unitary rotation of
elements of any local subset $S_{k}^{(i)}$, which, in a sense,
makes them independent. Hence, for the purposes of the
present framework, we propose to call it {\it local independence
property}.

A particular example of a set having the above property is the
already introduced Shifts UPB (\ref{upb3q}). At each site
there are four different vectors $\ket{0}$, $\ket{1}$, $\ket{e}$,
and $\ket{\overline{e}}$, which can be grouped in two distinct
sets $S_0=\{\ket{0},\ket{1}\}$ and
$S_1=\{\ket{e},\ket{\overline{e}}\}$. Since, by the
very assumption, $\ket{e}\neq \ket{0},\ket{1}$, none of the
vectors from $S_0$ is orthogonal to none of elements of $S_1$,
and hence $U_{\mathrm{Shifts}}$ has the local independence property.

Interestingly, as it can easily be checked, all sets of orthogonal
vectors in multi-qubit Hilbert spaces have the above property and
all local subsets contain at most two elements.
On the other hand, the example of TILES UPB shows that this in
general is not the case when local dimensions are larger than two.

Let us now pass to our construction of Bell inequalities. To every
vector $\ket{\Psi_m}$ from $S$ [cf. Eq. (\ref{S})] we can
associate a conditional probability $P({\bf a}_m|{\bf x}_m)$, or,
strictly speaking, vectors of measurements settings and outcomes
\begin{equation}
{\bf a}_m=(a_m^{(1)},\ldots,a_m^{(N)})\qquad  \mathrm{and}\qquad
{\bf x}_m=(x_m^{(1)},\ldots,x_m^{(N)})
\end{equation}
in the following way:
\begin{itemize}
\item the measurement setting $x_m^{(i)}$ of the observer $i$ is given by the index $k$ enumerating
the subset $S_k^{(i)}$ containing $\ket{\psi_m^{(i)}}$,

\item the measurement outcome $a_m^{(i)}$ corresponds to the position of $\ket{\psi_m^{(i)}}$ in the set $S_k^{(i)}$.
\end{itemize}

Eventually, we simply add the obtained conditional probabilities
and maximize the resulting expression over all classical
correlations, which leads us to the following Bell inequality
\begin{equation}\label{BellIneq}
\sum_{m=1}^{|S|}P({\bf a}_m|{\bf x}_m)\leq 1.
\end{equation}
The value of the right-hand side of the above, the so-called
classical bound, directly follows from the orthogonality of
elements of $S$. Since the latter are product, for each pair of
vectors $\ket{\Psi_m}$, $\ket{\Psi_n}$ $(m\neq n)$, there exists
site, say $i$, such that
$\ket{\psi_m^{(i)}}\perp\ket{\psi^{(i)}_n}$, and so the local
independence property says implies
$\ket{\psi_m^{(i)}},\ket{\psi^{(i)}_n}$ are distinct elements of
the same local subset $S^{(i)}_k$. Consequently, the associated
conditional probabilities $P({\bf a}_m|{\bf x}_m)$ and $P({\bf
a}_n|{\bf x}_n)$ have at site $i$ the same measurement settings
but different outcomes. This means that for any deterministic
local model, if one of these two probabilities is one, the other
one equals zero. Let us further call such probabilities
orthogonal. Since the above holds for any pair of conditional
probabilities, the right-hand side of (\ref{BellIneq}) clearly
amounts to one.

Notice then that, in principle, we can consider a more general
inequality by combining the conditional probabilities $P({\bf
a}_m|{\bf x}_m)$ $(m=1,\ldots,|S|)$ with arbitrary positive
weights $q_m$. However, we always get
in this way a Bell inequality which is weaker that the one above
and certainly cannot be tight (see below).

\subsubsection{Properties}

Let us now shortly characterize the obtained Bell inequalities
(\ref{BellIneq}). We collect their most important properties in
the following theorem \cite{BellUPB1,BellUPB2}.

\begin{theorem}\label{theorem}
Let $S$ be a set of orthogonal product vectors from $H$ having the
local independence property. Then the following implications are
true:
\begin{description}
    \item[(i)] the associated Bell inequality (\ref{BellIneq}) is
    not violated by quantum correlations\\

    \item[(ii)] if $S$ is a UPB in $H$, then the Bell inequality
    (\ref{BellIneq}) is nontrivial in the sense that it is
    violated by some nonsignalling correlations,\\

    \item[(iii)] if $S$ is a full basis in $H$ or can be completed to
    one in such a way that it maintains the local independence
    property, the associated Bell inequality (\ref{BellIneq}) is
not violated by any nonsignalling correlations.
\end{description}
\end{theorem}

\begin{proof}(i): Let us, in contrary, assume that indeed the Bell
inequality (\ref{BellIneq}) associated to $S$ is violated by a
quantum state $\varrho$. Then, there exist local measurement
operators and the resulting Bell operator, denoted $B$, such that
$\tr(B\varrho)>1$. This means that at least one of the eigenvalues
of $B$ has to exceed one. On the other hand, it is clear that the
local measurement operators can be assumed to be projective; if
$\varrho$ violates (\ref{BellIneq}) with POVM, one is able to find
another quantum state $\varrho'$ acting on a larger Hilbert space
violating the same Bell inequality with projective measurements.

Let then $P_m=\otimes_{i=1}^{N}P_m^{(i)}$ denote a product
projective measurement operator corresponding to $P({\bf a}_m|{\bf
x}_m)$, which, in general, may be different from the corresponding
vectors $\ket{\Psi_m}\in S$. Clearly, orthogonality of the
conditional probabilities $P({\bf a}_m|{\bf x}_m)$ is translated
to the orthogonality of the corresponding $P_m$. Precisely, as
already stated, any pair of probabilities $P({\bf a}_m|{\bf x}_m)$
and $P({\bf a}_n|{\bf x}_n)$ has at some site, say $i$, the same
settings but different outcomes, implying that $P_{m}^{(i)}\perp
P_{n}^{(i)}$ and hence $P_m\perp P_n$. As a result all $P_m$
$(m=1,\ldots,|S|)$ are orthogonal and the Bell operator $B=\sum_m
P_m$ is again a projector contradicting the fact that for some
$\varrho$, $\tr(B\varrho)>1$.

(ii): Our proof is constructive, that is, for any Bell inequality
associated to a UPB we will provide particular NC violating it.
We denote by $\Pi$ the projector onto $\mathrm{span} S$,
and introduce, in a full analogy to (\ref{witness}), the following indecomposable witness
\begin{equation}\label{witness3}
W=\frac{1}{|S|-\dim H}(\Pi-\epsilon \mathbbm{1} )
\end{equation}
with $\epsilon$ being a positive number defined as
\begin{equation}\label{eps}
\epsilon=\min\langle x_1,\ldots,x_N|\Pi|x_1,\ldots,x_N\rangle,
\end{equation}
where minimum is taken over all fully product vectors from $H$.
One directly checks that this witness detects entanglement of the
state (\ref{BE}) constructed from the UPB $S$, i.e.,
$\mathrm{tr}(W\varrho)<0$. This, after substituting the exact (\ref{BellIneq})
of $\varrho$, can be rewritten as
\begin{equation}
\mathrm{tr}(W\Pi)>1.
\end{equation}
Clearly, $\Pi$ can be seen as a Bell operator corresponding to our
Bell inequality associated to $S$. To complete
the proof of (ii) it suffices then to notice any local
measurements performed on any entanglement witness, in particular
(\ref{witness3}), give nonsignalling correlations (see e.g.
\cite{Universal,Barnum}).

(iii): Let us start from the case when $|S|<\dim H$ and assume
that $S$ can be completed to a basis of $H$ maintaining the local
independence property (if $H=(\mathbbm{C}^2)^{\ot N}$ one can
always do that provided $S$ is completable). Let then
$\ket{\Psi_m}$ $(m=|S|+1,\ldots,\dim H)$ denote product orthogonal
vectors completing $S$, i.e.,
$\mathrm{span}(S\cup\{\ket{\Psi_m}\}_{m})=H$. Consequently, one
can associate a Bell inequality (\ref{BellIneq}) to the set $S$
and conditional probabilities $P({\bf a}_m|{\bf x}_m)$ to the new
vectors $\ket{\Psi_m}$ $(m=|S|+1,\ldots,\dim H)$ in an unambiguous
way. Then
\begin{equation}
\sum_{m=1}^{|S|}P({\bf a}_m|{\bf x}_m)\leq \sum_{m=1}^{\dim
H}P({\bf a}_m|{\bf x}_m)\leq 1,
\end{equation}
meaning that it suffices to prove that the Bell inequality
appearing on the right-hand side (the one constructed from a full
basis in $H$) is trivial. For this purpose, we note
that the latter is saturated by the uniform probability
distribution $P({\bf a}|{\bf x})=1/\dim H$ for any ${\bf a}$ and
${\bf x}$, which is an interior point of the corresponding
polytope of classical correlations. Consequently, this Bell
inequality is saturated by all vertices of the polytope,
and hence by any affine combination thereof, in particular, all
nonsignalling correlations. $\blacksquare$
\end{proof}

It is illuminating to see how the properties of $S$ determine the
properties of the associated Bell inequality. Orthogonality of
elements of $S$ implies that it lacks quantum violation. If $S$ is
additionally a UPB, then the Bell inequality is nontrivial because
it detects some nonsignalling correlations. On the other hand,
pit is trivial if $S$ is a full basis in $H$ or can be
completed to one maintaining the local independence property. In
the case of $H=(\mathbbm{C}^{2})^{\ot N}$, up to sets that can
only be completed to UPBs, the implication (iii) becomes
equivalence \cite{BellUPB2}. In the higher-dimensional case,
however, there are sets having local independence property which
are not UPBs but cannot be extended maintaining the local
independence property (see Sec. \ref{sec:further}).

The more important and interesting question concerns the tightness
of these Bell inequalities. As shown in Refs.
\cite{BellUPB1,BellUPB2} there exist example of both tight and
nontight Bell inequalities associated to UPBs (see Sec. \ref{sec:examples}
for examples) and, so far, it remains unclear what decides on
tightness.

\subsubsection{Examples}
\label{sec:examples}

Just to get a better insight into the construction let us apply to
it to some particular examples of sets $S$, in particular those
presented in Sec. \ref{UPB}.

\begin{example}Using the already exploited relation
between the GYNI Bell inequality (\ref{3player}) and Shifts UPB let us show how
the above construction works in practice. As already noticed,
$U_{\mathrm{Shifts}}$ has two different bases at each site
$S_0=\{\ket{0},\ket{1}\}$ and
$S_1=\{\ket{e},\ket{\overline{e}}\}$. The vector
$\ket{e}\in\mathbbm{C}^2$ is, by assumption, different than
$\ket{0}$ and $\ket{1}$, and hence $U_{\mathrm{Shifts}}$ has the
local independence property. We then associate a conditional
probability to every vector in $U_{\mathrm{Shifts}}$:
\begin{eqnarray}
  &\ket{000} \mapsto P(000|000), \qquad \qquad
  \ket{1\overline{e}e} \mapsto P(110|011),& \nonumber\\
&  \ket{e1\overline{e}} \mapsto P(011|101),\qquad \qquad
  \ket{\overline{e}e1} \mapsto P(101|110).&
\end{eqnarray}
Simply by adding the above probabilities we get (\ref{3player}).
In exactly the same was one shows that GYNI$_N$ can be associated
to a certain $N$-qubit UPB \cite{BellUPB1,BellUPB2}. Moreover, it
was recently shown in Ref. \cite{BellUPB2} that GYNI$_N$ is tight
for odd $N$.

Interestingly, the GYNI$_3$ is the only tight three-partite Bell
inequality with no quantum violation in the scenario of two
dichotomic measurements per site, and it is associated to the only
class of UPB in $(\mathbbm{C}^2)^{\ot 3}$ characterized in Ref.
\cite{Bravyi}.

\end{example}

\begin{example}Second, let us consider the Generalized Shifts UPB
(\ref{GenShifts}). The corresponding Hilbert space is
$H=(\mathbbm{C}^{2})^{\ot N}$ with $N=2k-1$ for integer $k\geq 2$.
Following the above rules, at each site one can define $k$ local
subsets $S_0=\{\ket{0},\ket{1}\}$ and
$S_i=\{\ket{e_i},\ket{\overline{e}_i}\}$ $(i=1,\ldots,k-1)$, which
will later define $k$ observables. We then associate a conditional
probability to every element of $U_{\mathrm{GenShifts}}$:
\begin{eqnarray}
\ket{0\ldots0}&\mapsto& P(0\ldots0|0\ldots 0)\nonumber\\
\ket{1e_1\ldots e_{k-1}\overline{e}_{k-1}\ldots \overline{e}_{1}}&\mapsto& P(10\ldots01\ldots 1|01\ldots k-1,k-1\ldots 1)\nonumber\\
&\vdots&\nonumber\\
\ket{e_1\ldots e_{k-1}\overline{e}_{k-11}\ldots
\overline{e}_{1}}&\mapsto& P(0\ldots01\ldots 11|01\ldots
k-1,k-1\ldots 11).
\end{eqnarray}
Summing all these probabilities up, we get the $N$-partite Bell
inequality with odd $N$:
\begin{equation}\label{BellIneq2}
P(0\ldots 0|0\ldots 0)+\sum_{i=1}^{2k-1}D^{i}P(10\ldots01\ldots
1|01\ldots k-1,k-1\ldots 1)\leq 1,
\end{equation}
where $D$ denotes an operation shifting the input and output
vectors by one to the right, i.e.,
$D(x_1,\ldots,x_N)=(x_N,x_1,\ldots,x_{N-1})$. Notice that since at
each site one has $k$ two-element local subsets $S_i$, the Bell
inequality (\ref{BellIneq2}) corresponds to the scenario with $k$
dichotomic observables per site.

Due to theorem \ref{theorem}, all the Bell inequalities
(\ref{BellIneq2}) are nontrivial. However, it is unclear whether
they are tight. For $N=3$ the above class recovers the GYNI$_3$
which is tight, while already for $N=5$ the corresponding Bell
inequality is not tight.
\end{example}

\begin{example}Consider now the class of UPBs provided in Ref.
\cite{NisetCerf}, i.e., $U_{\mathrm{NC}}$ presented in example 3.
Here $H=(\mathbbm{C}^d)^{\ot N}$ with $d\geq N-1$. From Eqs.
(\ref{NC1}) and (\ref{NC2}) it follows that at each site one can
distinguish two local subsets $S_{0}=\{\ket{i}\}_{i=0}^{d-1}$,
i.e., the standard basis, and $S_1=\{\ket{e_i}\}_{i=0}^{d-1}$.
Since the elements of $U_{\mathrm{NC}}$ are orthogonal
irrespectively of the choice of the second basis, $U_{NC}$ has the
local independence property. Associating conditional probabilities
to elements of $U_{\mathrm{NC}}$ and summing them up, one gets the
$N$-partite Bell inequality:
\begin{eqnarray}\label{BellIneq3}
P(d-1,\ldots,d-1|1,\ldots,1)+\sum_{i=0}^{N-1}\sum_{j=0}^{d-2}D^{i}P(0,1,\ldots,d-1,j|0,\ldots,0,1)\leq
1,
\end{eqnarray}
where $D$ is defined as before and $D^{0}$ is an identity.

Theorem \ref{theorem} says that all the Bell inequalities
(\ref{BellIneq3}) are nontrivial, however, it is not clear
whether they are tight in general. For $N=3$ and $d=2$, this class
gives GYNI$_3$, but for $N=4$ and $d=3$ one checks that the
resulting Bell inequality is not tight.

\end{example}

Let us notice that within the above framework one can also obtain
tight Bell inequalities with no quantum violation from UPBs that
are independent of GYNI$_N$. Some new examples as for instance the
following four-partite Bell inequality
\begin{eqnarray}
p(0000|0000) + p(1000|0111) + p(0110|1012) +
p(0001|0110)\nonumber\\ + p(1011|0001) + p(1101|0102) +
p(1110|1101) \leq 1
\end{eqnarray}
were found recently in Ref. \cite{BellUPB2}.

\subsection{Further generalizations}
\label{sec:further}

We will show now that not only UPBs lead to nontrivial Bell
inequalities with no quantum violation. If at least one $d_i$ in
$H$ is larger than two, there exist sets of orthogonal product
vectors that are not UPBs in the sense of definition
\ref{def:upb}, but still, the associated Bell inequalities
(\ref{BellIneq}), via the rules from Sec. \ref{sec:constr}, lack
quantum violation and are violated by nonsignalling correlations.

To be more precise, let us consider again set of orthogonal product vectors
$S$ and let us split local sets $S^{(i)}$ [cf. Eq. (\ref{Si})] into
subsets $S^{(i)}_k$ following the same rules as before. Then, we have the definition.

\begin{definition}\label{def:wupb}
Let $S$ be a set of orthogonal fully product vectors from $H$
having the local independence property. Then, if $|S|<\dim H$ and
there does not exist a product vector
$\ot_{i=1}^{N}\ket{\phi_i}\in H$ with $\ket{\phi_i}\in S^{(i)}$
which is orthogonal to all vectors from $S$, we call $S$ a {\it
weak unextendible product basis} (wUPB).
\end{definition}

Clearly, any UPB is also a wUPB. Also, if all $d_i=2$ in eq. (\ref{HilbertSpace}),
these two notions are equivalent. If, however, at least one of the local
dimensions $d_i$ is larger than two, there exist wUPB that are not UPB. As a particular example
consider the following.

\begin{example}Consider the following set of vectors from $H=\mathbbm{C}^2\ot\mathbbm{C}^2\ot\mathbbm{C}^3$:
\begin{equation}\label{setwUPB}
S=\{\ket{000},\ket{1\overline{e}f},\ket{e1\overline{f}},\ket{\overline{e}e1},\ket{\overline{e}e2},\ket{e1\hat{f}}\},
\end{equation}
where $\ket{f},\ket{\overline{f}}$, and $\ket{\hat{f}}$ are three orthogonal vectors from $\mathbbm{C}^3$.
At the first two sites one distinguishes two local sets $S_0^{(1)}=S_0^{(2)}=\{\ket{0},\ket{1}\}$
and $S_1^{(1)}=S_1^{(2)}=\{\ket{e},\ket{\overline{e}}\}$, while at the third site
$S^{(3)}_0=\{\ket{0},\ket{1},\ket{2}\}$ and $S^{(3)}_1=\{\ket{f},\ket{\overline{f}},\ket{\hat{f}}\}$.

The set $S$ has the local independence property because irrespectively of the choice
of all these subsets, all its elements are orthogonal. However, it is clearly not a
UPB because $\ket{e0g}$ and $\ket{\overline{e}\overline{e}g}$ with
$\mathbbm{C}^3\ni\ket{g}\perp\ket{0},\ket{f}$ are orthogonal to $S$. Still, $S$ is
a wUPB; there is no product vector $\ket{\phi_1}\ket{\phi_2}\ket{\phi_3}\in\mathbbm{C}^2\ot\mathbbm{C}^2\ot\mathbbm{C}^3$
with $\ket{\phi_i}\in S_j^{(i)}$ $(i=1,2,3;j=1,2)$, which is orthogonal to $S$.

\end{example}

It remains an open question as to whether the quantum state constructed from
a wUPB, i.e., the state (\ref{BE}) with $\Pi$ denoting now a projector
onto the subspace spanned by the wUPB, is entangled. It is, nevertheless, still a PPT
state.

Following the rules given in Sec. \ref{sec:constr}, any wUPB can be
associated a Bell inequality (\ref{BellIneq}) with no quantum violation
which is violated by some nonsignalling correlations. In fact,
we have the following theorem.

\begin{theorem}\label{theorem2}
If $S$ is a wUPB, the associated Bell inequality
(\ref{BellIneq}) is violated by some nonsignalling correlations.
\end{theorem}

\begin{proof}The proof goes along the same lines as the one of (ii) of
theorem \ref{theorem}. It suffices to consider the same
operator as in Eq. (\ref{witness3}) with $\Pi$ denoting now a projector
onto the subspace spanned by the wUPB $S$
and the minimum in Eq. (\ref{eps}) taken over product vectors $\ot_{i=1}^N\ket{\phi_i}\in H$
with local vectors $\ket{\phi_i}\in S^{(i)}$.

A measurements of such $W$ along the settings corresponding to the
local sets $S_{k}^{(i)}$ produce the value of the Bell inequality
(\ref{BellIneq}) constructed from the wUPB $S$ given by
$|S|(1-\epsilon)/(|S|-\epsilon\dim H)$. This, due to the fact that
$|S|<\dim H$ and $\epsilon<|S|/\dim H$, is always larger than one.
$\blacksquare$
\end{proof}

Notice that if $S$ is a wUPB but not UPB, the operator $W$ is no longer
an entanglement witness, but still it is a Hermitian operator. It therefore
represents nonsignalling correlations \cite{Universal}, which are
not Gleason correlations. It is then an open question as to whether
Bell inequalities associated to wUPBs "detect" Gleason correlations.

To conclude, let us notice that the Bell inequality corresponding to
the set (\ref{setwUPB}):
\begin{eqnarray}
&&p(000|000)+p(110|011)+p(011|101)+p(101|110)\nonumber\\
&&+p(012|101)+p(102|110)\leq 1,
\end{eqnarray}
which has two three-outcome observables at the third site, is tight.
This is because it is lifted three-partite GYNI Bell inequality (\ref{3player}) \cite{Stefano}.

\section{Conclusions}
\label{sec:conclusion}

'Guess you neighbour's input' is a multipartite nonlocal game
that, despite its simplicity, captures important features of
multipartite correlations. Moreover, it has unexpected connections
to topics in quantum foundations and quantum information theory.
In particular, it shows that the natural multipartite
generalization of Gleason's Theorem fails for more than two
parties, that intrinsically multipartite principles are needed to
characterize quantum correlations and that there exists a link
between unextendible orthogonal product bases and Bell
inequalities with no quantum violation.

From a speculative point of view, GYNI suggests that we are
lacking an intrinsically multipartite principle in our
understanding of correlations. Indeed, the most interesting
feature of the game is that it represents a multipartite
strengthening of the no-signaling principle, which is by
construction a bipartite principle, that is obeyed by quantum
correlations. This naturally raises the question of what physical
or information-theoretic principles lie behind GYNI. We hope our
work stimulates further research in this direction.

\begin{acknowledgement}
Discussions with T. Fritz are acknowledged. This work was
supported by the ERC starting grant PERCENT, the EU AQUTE and QCS
projects, the Spanish CHIST-ERA DIQIP, FIS2008-00784 and
FIS2010-14830 projects, and the UK EPSRC. R. A. is supported by
the Spanish MINCIN through the Juan de la Cierva program.
\end{acknowledgement}

\end{document}